\documentclass[draftclsnofoot,12pt,onecolumn]{IEEEtran}

\usepackage{amsthm}
\usepackage{graphicx}
\usepackage{amsmath}
\usepackage{amssymb}
\usepackage{cases}
\usepackage{latexsym}
\usepackage{multirow}
\usepackage{array}
\usepackage{subfigure}
\usepackage[dvips]{epsfig}
\theoremstyle{theorem}
\newtheorem{lemma}{Lemma}
\newtheorem{theorem}{Theorem}
\newtheorem{corollary}{Corollary}

\begin{document}
\title{Heterogeneous Cellular Networks with Flexible Cell Association: A Comprehensive Downlink SINR Analysis}
\author{Han-Shin Jo, Young Jin Sang, Ping Xia, and Jeffrey G. Andrews
\thanks{H. S. Jo, P. Xia, and J. G. Andrews are with Department of Electrical and Computer Engineering, The University of Texas at Austin, USA
(e-mail: han-shin.jo@austin.utexas.edu, pxia@mail.utexas.edu and jandrews@ece.utexas.edu). Young Jin Sang is with Department of Electrical and Electronic Engineering, Yonsei University, Korea. (e-mail: yjmich@dcl.yonsei.ac.kr). This research was supported by Motorola Solutions, Arlington Heights, IL, and the USA's National Science Foundation, CIF-1016649. Manuscript last updated: \today}}
\maketitle
\begin{abstract}
In this paper we develop a tractable framework for SINR analysis in downlink heterogeneous cellular networks (HCNs) with flexible cell association policies. The HCN is modeled as a multi-tier cellular network where each tier's base stations (BSs) are randomly located and have a particular transmit power, path loss exponent, spatial density, and bias towards admitting mobile users.  For example, as compared to macrocells, picocells would usually have lower transmit power, higher path loss exponent (lower antennas), higher spatial density (many picocells per macrocell), and a positive bias so that macrocell users are actively encouraged to use the more lightly loaded picocells.  In the present paper we implicitly assume all base stations have full queues; future work should relax this.  For this model, we derive the outage probability of a typical user in the whole network or a certain tier, which is equivalently the downlink SINR cumulative distribution function. The results are accurate for all SINRs, and their expressions admit quite simple closed-forms in some plausible special cases. We also derive the \emph{average ergodic rate} of the typical user, and the \emph{minimum average user throughput} -- the smallest value among the average user throughputs supported by one cell in each tier. We observe that neither the number of BSs or tiers changes the outage probability or average ergodic rate in an interference-limited full-loaded HCN with unbiased cell association (no biasing), and observe how biasing alters the various metrics.
\end{abstract}
\IEEEpeerreviewmaketitle

\section{Introduction}\label{sec:intro}
Heterogeneous cellular networks (HCNs) comprise a conventional cellular network overlaid with a diverse set of lower-power base stations (BSs) such as picocells \cite{Pico97}, femtocells \cite{ChaAndGat08,JoYook09,JoYook10}, and perhaps relay BSs \cite{Relay09}.  Heterogeneity is expected to be a key feature of 4G cellular networks, and an essential means for providing higher end-user throughput \cite{ParDah11,QcommHetNet} as well as expanded indoor and cell-edge coverage. The \emph{tiers} of BSs are ordered by transmit power with tier 1 having the highest power.  Due to differences in deployment, they also in general will have differing path loss exponents and spatial density (e.g. the number of BSs per square kilometer).  Finally, in order to provide relief to the macrocell network -- which is and will continue to be the main bottleneck -- lower tier base stations are expected to be designed to have a bias towards admitting users \cite{ParDah11}, since their smaller coverage area usually results in a lighter load. For example, as shown in Fig. \ref{fig:BiasingView}, a picocell may claim a user even though the macrocell signal is stronger to the user.  The goal of this paper is to propose and develop a model and analytical framework that successfully characterizes the signal-to-noise-plus-interference ratio (SINR) -- and its derivative metrics like outage/coverage and data rate -- in such a HCN with arbitrary per-tier association biases.

\subsection{Motivation and Related Work}
The SINR statistics over a network are, unsurprisingly, largely determined by the locations of the base stations (BSs).  These locations are usually unknown during the design of standards or even a specific system, and even if they are known they vary significantly from one city to the next.  Since the main aspects of the system must work across a wide variety of base station deployments, the BS locations are usually abstracted to some baseline model. The most popular such abstraction is the hexagonal grid model, which has been used extensively both in industry \cite{ChiLin04,JaiTam08,IkuRup10} and academia \cite{WuSze99,Kar99}, over at least the past three decades and it remains ubiquitous up to the present day. The SINR expressions resulting from such a model are complex and depend on multiple random variables, and are so metrics of interest are usually estimated by Monte Carlo methods. This complicates the understanding of how the various system parameters affect such metrics. Furthermore, given the previously mentioned deployment trends, the continued relevance of the hexagonal grid model is open to debate.

An alternative is to model the BS locations as random and drawn from a spatial stochastic process, such as the Poison point process (PPP).  Such an approach has been advocated as early as in 1997 \cite{Fleming97,BacCellular,Bro00} even for traditional coverage-focused deployments. Recent work has found the complete downlink SINR distribution under fairly general assumptions, and shown that at least in one fairly regular macrocell deployment, that it is about as accurate as the grid model \cite{AndBac10}.  In the interim, the PPP model has been used extensively for modeling unplanned networks such as femtocells \cite{ChaAnd09a,ChaAnd09b,ChaKou09} or ad hoc networks \cite{BacNOW,HaeAnd09}; and for both it is now fairly well-accepted \cite{AndGan10,WinPin09}.

In the context of the heterogeneous deployments of interest in the present paper, a tractable model for SINR in a general $K$-tier downlink HCN has recently been found under certain assumptions \cite{DhiAnd11,DhiAndJ11}. The results are exact for positive SINR (in dB) only, do not include biasing, and the per-tier coverage probabilities are not given.   Clearly, it is desirable to develop a more general analytical model which is accurate for all SINRs allows for a flexible cell association.  This would allow for the consideration of biasing in the overall system design, which is considered a very important factor in the overall performance of an HCN \cite{BiasAlcatel10,BiasHuawei10,BiasNTT10,Sal_PIMRC10}.  We also note that independently, \cite{AndBac10} has been extended to the case of an unbiased two-tier network, where both macro BSs and femto BSs are located according to a PPP \cite{Muk11a}.

\subsection{Contributions}
We present a novel analytical model and set of results for SINR in downlink HCNs with flexible cell association.  The model characterizes the $K$ tiers of a cellular network by the transmit power, BS spatial density, path loss exponent, and bias factor. We assume that in addition to path loss that there is Rayleigh fading\footnote{Lognormal shadowing on both the desired and interference signals was considered in \cite{AndBac10} for one-tier, but significantly degrades tractability while not changing the main trends. We neglect it here for expediency, but it could be handled in a straightforward manner following the approach of \cite{AndBac10}.}. Unlike \cite{DhiAnd11,DhiAndJ11} where each mobile user connects to the BS offering the highest \emph{instantaneous} SINR, we assume the user connects to the BS that offers the maximum \emph{long-term averaged} received power (with biasing), i.e. fading is ignored. This results in a completely new approach to modeling and analysis of the HCNs. Specifically, we provide the following theoretical contributions.

First, we derive a pair of prerequisite quantities that are of interest in their own right. These are the (i) \emph{per-tier association probability}, which is the probability that a typical user is associated with each tier,  , and (ii) the \emph{distance to the serving BS} for a given bias. These are both probabilistic functions. From the per-tier association probability, the average number of users associated with a BS in each tier is derived, which captures the effect of cell association on the \emph{cell load} of each tier.

Next, we derive the complete \emph{outage probability} over all SINRs, with arbitrary biasing. The outage probability relative to an SINR threshold is equivalently the CDF of SINR for a randomly selected mobile in the network (or a certain tier). Although the results are not closed-form for the most general case, they are easily computable. In addition, they reduce to closed-form in several cases, for example an interference-limited network with the same per-tier path loss exponents and no biasing. In that particular case, the outage probability does not depend on the spatial density, transmit power or even the number of tiers.  Intuitively, adding more infrastructure of any type has a negligible affect on the coverage or outage, so primarily serves to provide more throughput (since more users can be served simultaneously as the number of base stations increases).

Finally, we derive two measures of spectral efficiency in our proposed HCN model: the \emph{average ergodic rate} of a randomly chosen user in the whole network (or a certain tier), and the \emph{minimum average user throughput} -- the smallest value among the average user throughputs supported by one cell in each tier. The average ergodic rate, like the outage probability, is independent of BS transmit power, BS density, and the number of tiers $K$ in interference-limited HCN without biasing, whereas it is strictly worsened by biasing in a fully-loaded (full queue) network. However, rather than drawing conclusions about the relative merits of different bias values, the main contribution of this paper is to provide a tractable SINR framework for doing so in future work, since it seems clear that the loading on each tier will significantly affect the cell association policies.

\section{Downlink System Model}\label{sec:model}
A fairly general model of an HCN would include $K$ tiers of BSs that are distinguished by their spatial densities, transmit powers, path loss exponents, and biasing factors. For example, a high-power macrocell BS network is overlaid with successively denser and lower power picocells (with positive biasing) and femtocells (also with positive biasing) in a plausible $K=3$ scenario as seen in Fig. \ref{fig:Biasing}.
Femto BS locations will generally be unplanned and so are well-modeled by a spatial random process \cite{ChaAnd09b,ChaKou09,Muk11a,KimHon10b}.  Perhaps more surprisingly, macrocell BSs may also be reasonably well-modeled by a random spatial point process, with about the same as accuracy as the typical grid model \cite{AndBac10,BacCellular,Fleming97}. Under this model, the positions of BSs in the $j$th tier are modeled according to a homogeneous PPP $\Phi_j$ with intensity $\lambda_j$ and users also are located according to a homogeneous PPP $\Phi^{(\mathrm{u})}$ with intensity $\lambda^{(\mathrm{u})}$ that is independent of $\{\Phi_j\}_{j=1, \cdots, K}$.

The downlink desired and interference signals experience path loss, and for each tier we allow different path loss exponents $\{\alpha_j\}_{j=1,\cdots,K}>2$, which is permissible in consideration with empirical value in cellular network \cite{3GPP36814}. Rayleigh fading with unit average power models the random channel fluctuations. No intra-cell interference is considered, i.e. orthogonal multiple access is employed within a cell. The noise is additive and has constant power $W$ but no specific distribution is assumed. Every BSs in the $j$th tier uses the same transmit power $\{P_j\}_{j=1,\cdots,K}$.

Denote $k \in \{1,\cdots,K\}$ as the index of the tier with which a typical user associated.
To evaluate the outage probability and achievable rate, we shift all point process so that a typical user (receiver) lies at the origin. Despite the shift, the distribution of potential interfering BS is still a homogeneous PPP with the same intensity. We thus denote $|Y_{ki}|$ as the distance from BS $i\in \Phi_k$ to the origin, i.e. the typical user. Also denote $\{R_j\}_{j=1,\cdots,K}$ as a distance of the typical user from the nearest BS in the $j$th tier, which is a main parameter for the proof following Lemma 1 and 3.

\subsection{Flexible Cell Association and Cell Load}
We assume open access, which means a user is allowed to access any tier's BSs.  This may not be true in some cases (notably femtocells) but provides best-case coverage, and we leave extension to closed-access to future work. We consider a cell association based on maximum biased-received-power (BRP) (termed {\it biased association}), where a mobile user is associated with the strongest BS in terms of BRP at the user. The BRP $\{P_{\mathrm{r},j}\}_{j=1,\cdots,K}$ is
\begin{equation}\label{eq:Pr_j}
P_{\mathrm{r},j} = P_j L_0 (R_j/r_0)^{-\alpha_j}B_j,
\end{equation}
where $L_0$ is the path loss at a reference distance $r_0$ (typically about $(4\pi/\nu)^{-2}$ for $r_0=1$, where $\nu$ denotes the wavelength). All BSs in the $j$th tier adopt identical bias factor $B_j$, which is positive value. When $\{B_j\}_{j=1,\cdots,K}=1$, i.e. no biasing, biased association is the same as conventional cell association (termed {\it unbiased association}), where the user connects to the BS that offers the strongest average power to the user. In particular, when $B_j=1/P_j$ for all $j=1,\cdots,K$, the user connects to the BS with the lowest path loss.
Employing $B_j>1$ extends the cell range (or coverage) of the $j$th tier BSs as shown in Fig. \ref{fig:Biasing}. This allows the tier selection to be tuned for cell load balancing or other purposes.

For clarity of exposition, we define
\begin{align}
\widehat{P}_j\triangleq\frac{P_j}{P_k}, ~\widehat{B}_j \triangleq\frac{B_j}{B_k}, ~\widehat{\alpha}_j \triangleq\frac{\alpha_j}{\alpha_k},
\end{align}
which respectively characterizes transmit power ratio, bias power ratio and path loss exponent ratio of interfering to serving BS. Note that $\widehat{P}_k=\widehat{B}_k=\widehat{\alpha}_k=1$.

Each tier's BS density and transmit power as well as cell association determine the probability that a typical user is associated with a tier. The following lemma provides the \emph{per-tier association probability}, which is essential for deriving the main results in the sequel.
\begin{lemma}
The probability that a typical user is associated with the $k$th tier is
\begin{align}\label{eq:Lem1}
\mathcal{A}_k = 2\pi \lambda_k \int_0^{\infty} r \exp \left\{-\pi \textstyle\sum_{j=1}^{K} \lambda_j (\widehat{P}_j \widehat{B}_j)^{2/\alpha_j} r^{2/\widehat{\alpha}_j} \right \} \mathrm{d}r.
\end{align}
If $\{\alpha_j\}=\alpha$, the probability is simplified to
\begin{align}\label{eq:Lem1-0}
\mathcal{A}_k = \left(1+ \frac{\sum_{j=1,j\neq k}^{K} \lambda_j (P_j B_j)^{2/\alpha}} {\lambda_k (P_k B_k)^{2/\alpha}} \right)^{-1} .
\end{align}
\end{lemma}
\begin{proof}See Appendix~\ref{sec:PFLem1}.
\end{proof}
Lemma 1 confirms the common sense that a user prefers to connect to a tier with higher BS density, transmit power, and bias. We further observe that BS density is more dominant in determining $\mathcal{A}_k$ than BS transmit power or bias factor because $\frac{2}{\alpha}<1$ ($\alpha>2$) both in practice and in our network model.
From Lemma 1, we easily derive the average number of users per BS in the $k$th tier, which quantifies the cell load of each tier.
\begin{lemma}
The average number of users associated with a BS in the $k$th tier (also the cell load of the $k$th tier) is given as
\begin{align}\label{eq:Cor1-1}
\mathcal{N}_k = 2\pi \lambda^{\mathrm{(u)}} \int_0^{\infty} r \exp \left\{-\pi \textstyle\sum_{j=1}^{K} \lambda_j (\widehat{P}_j \widehat{B}_j)^{2/\alpha_j} r^{2/\widehat{\alpha}_j} \right \} \mathrm{d}r.
\end{align}
If $\{\alpha_j\}=\alpha$, (\ref{eq:Cor1-1}) simplifies to
\begin{align}\label{eq:Cor1-2}
\mathcal{N}_k = \frac{\lambda^{\mathrm{(u)}}} {\lambda_k+\sum_{j=1,j\neq k}^{K} \lambda_j (\widehat{P}_j \widehat{B}_j)^{2/\alpha}}.
\end{align}
\end{lemma}
\begin{proof}
Denoting $S$, $N$, $N_{k}^{\mathrm{(b)}}$, and $N_{k}^{\mathrm{(u)}}$ as the area of the entire network, number of users in the entire network, average number of $k$th tier BSs and $k$th tier users, respectively, we obtain $N_{k}^{\mathrm{(u)}}=\mathcal{A}_k N = \mathcal{A}_k \lambda^{\mathrm{(u)}} S$ and $N_{k}^{\mathrm{(b)}}= \lambda_k S$. From the relations, the number of users per BS in the $k$th tier is given by
\begin{align}\label{eq:Cor1-3}
\mathcal{N}_k=\frac{N_{k}^{\mathrm{(u)}}}{N_{k}^{\mathrm{(b)}}}
=\frac{\mathcal{A}_k\lambda^{\mathrm{(u)}}}{\lambda_k}.
\end{align}
Combining (\ref{eq:Cor1-3}) with (\ref{eq:Lem1}) gives the desired result. In case of $\{\alpha_j\}=\alpha$, the desired result follows by using (\ref{eq:Lem1-4}) instead of (\ref{eq:Lem1}).
\end{proof}
Higher BS transmit power or biasing lead more users to connect with each BS of the corresponding tier.
Comparing Lemma 1 with Lemma 2 admits the perhaps surprising observation that as the BS density of the $k$th tier increases, more users are associated with the $k$th tier but the number of users per BS decreases even in that tier. The decrease in the number of users per BS is because in (\ref{eq:Cor1-3}), $A_{k}$ scales like $\frac{\lambda_k}{\lambda_k+c}$ less than $\lambda_k$ for a positive value of $c$ (see (\ref{eq:Lem1-4})). On the other hand, as more users are associated with the $k$th tier, less users are associated with other tiers. It implies that deploying more BSs reduces the cell load of the corresponding tier as well as other tiers, whereas employing higher transmit power or bias factor increases the cell load of the corresponding tier while decreasing that of other tiers.

\subsection{Statistical Distance to Serving Base Station}
We consider a typical user at the origin associated with the $k$th tier. Denote $X_k$ as the distance between the user and its serving BS. Since BSs are deployed as a PPP, $X_k$ is a random variable described by its probability density function (PDF) given in Lemma 3.
\begin{lemma}
The PDF $f_{X_k}(x)$ of the distance $X_k$ between a typical user and its serving BS is
\begin{align}\label{eq:Lem2}
f_{X_k}(x) = \frac{2\pi\lambda_k }{\mathcal{A}_k} x \exp \left\{-\pi\textstyle\sum_{j=1}^{K} \lambda_j (\widehat{P}_j \widehat{B}_j)^{2/\alpha_j} x^{2/\widehat{\alpha}_j} \right \}.
\end{align}
\end{lemma}
\begin{proof}
Since the event of $X_k>x$ is the event of $R_k>x$, given the typical user's association with the $k$th tier, the probability of $X_k>x$ can be given as
\begin{align}\label{eq:CCDF_X}
\mathbb{P}[X_k>x]&=\mathbb{P}[R_k>x \mid n=k]=\frac{\mathbb{P}[R_k>x, n=k]}{\mathbb{P}[n=k]},
\end{align}
where $\mathbb{P}[n=k]=\mathcal{A}_k$ follows from Lemma 1, and the joint probability of $R_k>x$ and $n=k$ is
\begin{align}\label{eq:CCDF_X1}
\mathbb{P}[R_k>x, n=k]&= \mathbb{P}\left[R_k>x, P_{\mathrm{r},k}(R_k)> \max_{j, j\neq k} P_{\mathrm{r},j}\right]\cr
&=\int_{x}^{\infty} \textstyle\prod_{j=1, j\neq k}^K \mathbb{P}\left[P_{\mathrm{r},k}(r) > P_{\mathrm{r},j}\right] f_{R_k}(r)\mathrm{d}r \cr
&\mathop =\limits^{\left( a \right)}
\int_{x}^{\infty} \textstyle\prod_{j=1, j\neq k}^K \mathbb{P}\left[R_j > (\widehat{P}_j\widehat{B}_j)^{1/\alpha_j}r^{1/\widehat{\alpha}_j} \right] f_{R_k}(r)\mathrm{d}r \cr
&\mathop =\limits^{\left( b \right)}
2\pi \lambda_k \int_x^{\infty} r \exp \left\{-\pi \textstyle\sum_{j=1}^{K} \lambda_j (\widehat{P}_j \widehat{B}_j)^{2/\alpha_j} r^{2/\widehat{\alpha}_j} \right\} \mathrm{d}r,
\end{align}
where $(a)$ follows from (\ref{eq:Pr_j}), and $(b)$ is given from (\ref{eq:Lem1-2}) and (\ref{eq:Lem1-3}).
Plugging (\ref{eq:CCDF_X1}) into (\ref{eq:CCDF_X}) gives
\begin{align}\label{eq:CCDF_X2}
\mathbb{P}[X_k>x]=\frac{2\pi \lambda_k}{\mathcal{A}_k} \int_x^{\infty} r \exp \left\{-\pi \textstyle\sum_{j=1}^{K} \lambda_j (\widehat{P}_j \widehat{B}_j)^{2/\alpha_j} r^{2/\widehat{\alpha}_j} \right\} \mathrm{d}r,
\end{align}
The CDF of $X_k$ is $F_{X_k}(x)=1-\mathbb{P}[X_k > x]$ and the PDF is given as
\begin{align*}
f_{X_k}(x)=\frac{\mathrm{d}F_{X_k}(x)}{\mathrm{d}x}= \frac{2\pi\lambda_k }{\mathcal{A}_k} x \exp \left\{-\pi\textstyle\sum_{j=1}^{K} \lambda_j (\widehat{P}_j \widehat{B}_j)^{2/\alpha_j} x^{2/\widehat{\alpha}_j} \right\}.
\end{align*}
\end{proof}
Given a typical user associated with the $k$th tier, the interference at the user at distance $X_k=x$ is the sum of aggregate interference from all BSs, which are outside a disc of radius $(\widehat{P}_j \widehat{B}_j) ^{1/\alpha_j} x^{\widehat{\alpha}_j}$ for all $j=1,\cdots,K$ (see (a) in (\ref{eq:CCDF_X1})).
Note that this gives a large difference in the interference model between a cellular network with open access and previous studies for wireless ad hoc networks and two-tier femtocell networks with closed access, which had no such lower limit of the minimum distance from interfering BSs.

\section{Outage Probability}
Define the outage probability $\mathcal{O}$ as the probability that the instantaneous SINR of a randomly located user is less than a target SINR. Since the typical user is associated with at most one tier, from the law of total probability, the probability is given as
\begin{align}\label{eq:Pout}
\mathcal{O}
= \sum_{k=1}^K \mathcal{O}_{k}\mathcal{A}_k,
\end{align}
where $\mathcal{A}_k$ is the per-tier association probability given in Lemma 1 and $\mathcal{O}_{k}$ is the outage probability of a typical user associated with $k$th tier. For a target SINR $\tau$ and a typical user $\mathtt{SINR}_k(x)$ at a distance $x$ from its associated BS, the outage probability is
\begin{align}\label{eq:PoutDef}
\mathcal{O}_{k}&\triangleq \mathbb{E}_{x}\left[\mathbb{P}\left[\mathtt{SINR}_k(x) \leq \tau \right]\right]
\end{align}
The metric is the outage probability averaged over cell coverage (defined by $f_{X_k}(x)$ in Lemma 2). It represents the average fraction of the cell area that is in {\it outage} at any time. The metric is also exactly the CDF of SINR over the entire network.

The SINR of a typical user at a random distance $x$ from its associated BS in tier $k$ is
\begin{equation}\label{eq:SINR}
\mathtt{SINR}_k(x) = \frac{P_k g_{k,0} x^{-\alpha_k}}{\sum_{j=1}^{K} \sum_{i \in \Phi_j \backslash B_{k0} } P_j h_{ji} |Y_{ji}|^{-\alpha_j}+W/L_0},
\end{equation}
where $g_{k,0}$ is the exponentially distributed channel power with unit mean from the serving BS, $|Y_{ji}|$ is the distance from BS $i \in \Phi_j\backslash B_{k0}$ (except the serving BS $B_{k0}$) to the origin, and $h_{ji}$ is the exponentially distributed channel power with unit mean from the $i$th interfering BS in the $j$th tier.

\subsection{General Case and Main Result}
We now provide our most general result for the network outage probability from which all other results in this section follow.
\begin{theorem}
The outage probability of a typical user associated with the $k$th tier is
\begin{align}\label{eq:Pk}
\mathcal{O}_{k}=1- \frac{2\pi\lambda_k}{\mathcal{A}_k} \int_0^{\infty}  x \exp \left\{- \tfrac{\tau }{\mathtt{SNR}}  -\pi \textstyle\sum_{j=1}^{K} C_{j} x^{2/\widehat{\alpha}_j} \right \} \mathrm{d}x,
\end{align}
where $\mathtt{SNR}=\frac{P_k L_0 x^{-\alpha_k}}{W}$ and
\begin{align*}
C_{j} = \lambda_j \widehat{P}_j^{2/\alpha_j} \left[\widehat{B}_j^{2/\alpha_j}+\mathcal{Z}(\tau, \alpha_j, \widehat{B}_j)\right],
\end{align*}
with $\mathcal{Z}(\tau, \alpha_j, \widehat{B}_j) = \frac{2\tau \widehat{B}_j^{2/\alpha_j-1}}{\alpha_j-2} ~_{2}F_1\left[1, 1\!-\!\tfrac{2}{\alpha_j}; 2\!-\!\tfrac{2}{\alpha_j};  -\tfrac{\tau}{\widehat{B}_j} \right]$,
and $_2F_1[\cdot]$ denotes the Gauss hypergeometric function.
It follows that the outage probability of a randomly chosen user is
\begin{align}\label{eq:P}
\mathcal{O}=1- \sum_{k=1}^{K} 2\pi\lambda_k \int_0^{\infty}  x \exp \left\{- \tfrac{\tau }{\mathtt{SNR}}  -\pi \textstyle\sum_{j=1}^{K} C_{j} x^{2/\widehat{\alpha}_j} \right \} \mathrm{d}x.
\end{align}
\end{theorem}
\begin{proof}See Appendix~\ref{sec:PFThm1}.
\end{proof}
Although Theorem 1 does not give a closed-from expression, the integral is fairly easy to compute. For some plausible cases, we obtain simplified results in the following section.

\subsection{Special Case: Interference-Limited Network}
Since the BS density is typically quite high in heterogeneous cellular networks, the interference power easily dominates thermal noise. Thermal noise can often therefore be neglected, i.e. $W=0$, as we do in the rest of this section.
\subsubsection{No Noise, Equal Path Loss Exponents $\{\alpha_j\}=\alpha$}
\begin{corollary}
When $\{\alpha_j\}=\alpha$, the outage probability of the $k$th tier and overall network is respectively,
\begin{align}\label{eq:Pk1}
\mathcal{O}_{k}(\tau, \{\lambda_j\}, \alpha, \{\widehat{B}_j\})=1-
\frac{\sum_{j=1}^{K}\lambda_j (\widehat{P}_j \widehat{B}_j)^{2/\alpha}} {\sum_{j=1}^{K} \lambda_j \widehat{P}_j^{2/\alpha} \left[\widehat{B}_j^{2/\alpha}+\mathcal{Z}(\tau, \alpha, \widehat{B}_j) \right]}
\end{align}
and
\begin{align}\label{eq:P1}
\mathcal{O}(\tau, \{\lambda_j\}, \alpha, \{\widehat{B}_j\})=1-
\sum_{k=1}^{K} \left\{ \sum_{j=1}^{K} \frac{\lambda_j}{ \lambda_k} \widehat{P}_j^{2/\alpha} \left[\widehat{B}_j^{2/\alpha}+\mathcal{Z}(\tau, \alpha, \widehat{B}_j)\right] \right\}^{-1}.
\end{align}
\end{corollary}
\begin{proof}
If $\{\alpha_j\}=\alpha$, then $\widehat{\alpha}_j=1$. From (\ref{eq:Pk}) we obtain
\begin{align*}%\label{eq:Pk3}
\mathcal{O}_{k}(\tau, \{\lambda_j\}, \alpha, \{\widehat{B}_j\})=1 - \frac{2\pi \lambda_k}{\mathcal{A}_k}   \int_0^{\infty}  x \exp\left\{  -\pi  x^2 \textstyle\sum_{j=1}^{K} \lambda_j \widehat{P}_j^{2/\alpha} \left[\widehat{B}_j^{2/\alpha}+\mathcal{Z}(\tau, \alpha, \widehat{B}_j)\right] \right\}   \mathrm{d}x.
\end{align*}
Employing a change of  variables $x^2=v$ and $\int_0^\infty e^{-Av}\mathrm{d}v=\frac{1}{A}$, and plugging (\ref{eq:Lem1-4}) give the result in (\ref{eq:Pk1}). Similarly (\ref{eq:P}) leads to the result in (\ref{eq:P1}).
\end{proof}
This expression is simple and practically closed-form, requiring only the computation (or lookup) of a $\mathcal{Z}(a,b,c)$ value.
When $\{\widehat{B}_j\}=1$, i.e. no biasing, the probabilities are further simplified.
\begin{corollary}
When $\{\alpha_j\}=\alpha$ and $\{\widehat{B}_j\}=1$, i.e. unbiased association, the outage probability of $k$th tier and the network outage probability are given as
\begin{align}\label{eq:Pk11}
\mathcal{O}(\tau, \{\lambda_j\}, \alpha, 1) = \mathcal{O}_{k}(\tau, \{\lambda_j\}, \alpha, 1)=1-\frac{1}{1+\mathcal{Z}(\tau, \alpha,1)}.
\end{align}
\end{corollary}
\begin{proof}
Plugging $\widehat{B}_j=1$ in (\ref{eq:Pk1}), we obtain
\begin{align*}%\label{eq:Pk13}
\mathcal{O}_{k}(\tau, \{\lambda_j\}, \alpha, 1)=1-\frac{1}{1+\mathcal{Z}(\tau, \alpha,1)},
\end{align*}
which is independent on the index $k$. From (\ref{eq:Pout}) and $\sum_{k=1}^K \mathcal{A}_k=1$, we obtain
$\mathcal{O} = \mathcal{O}_{k}\sum_{k=1}^K \mathcal{A}_k=\mathcal{O}_{k}$. This gives the result in (\ref{eq:Pk11}).
\end{proof}
This is much simpler than Corollary 1, in particular the outage probability is now independent of the BS transmit power $P_j$ and BS density $\lambda_j$ for unbiased association.  It is also independent of the number of tiers $K$. This means that an open access unbiased HCN can randomly add new infrastructure, and the outage probability is not increased. This perhaps counter-intuitive result makes sense when one realizes that a mobile user always associates with the strongest BS, and so the SINR statistics do not change as more BSs are added. Interference management appears to be required to enhance downlink quality, since heterogeneous deployments do not change the outage probability. This ``invariance'' property is also observed in a interference-limited HCN with a different cell association \cite{DhiAndJ11}, where the user connects to the BS offering the maximum \emph{instantaneous} SINR, whereas here the user connects to the BS offering the maximum \emph{average} received power.

\subsubsection{No Noise, Equal Path Loss Exponents $\{\alpha_j\}=4$}
When $\{\alpha_j\}=4$, the Gauss hypergeometric function in Corollary 1 and 2 collapses to a simple arctangent function denoted by $\mathrm{arctan}()$. This provides a yet simpler expression for the outage probability.
\begin{corollary}
If $\{\alpha_j\}=4$, the outage probability of $k$th tier and overall network is respectively,
\begin{align}\label{eq:Pk2}
\mathcal{O}_{k}(\tau, \{\lambda_j\}, 4, \{\widehat{B}_j\})= 1-\frac{{\textstyle\sum}_{j=1}^{K}\lambda_j \sqrt{\widehat{P}_j \widehat{B}_j}} {{\textstyle\sum}_{j=1}^{K}\lambda_j \sqrt{\widehat{P}_j \widehat{B}_j} (1+\sqrt{\tau/\widehat{B}_j} \arctan(\sqrt{\tau/\widehat{B}_j}))}
\end{align}
and
\begin{align}\label{eq:P2}
\mathcal{O}(\tau, \{\lambda_j\}, 4, \{\widehat{B}_j\})=1-
\sum_{k=1}^{K} \left\{\sum_{j=1}^{K} \frac{\lambda_j}{ \lambda_k} \sqrt{\widehat{P}_j \widehat{B}_j} \left(1+\sqrt{\tau/\widehat{B}_j} \arctan\left(\sqrt{\tau/\widehat{B}_j}\right) \right) \right\}^{-1}.
\end{align}
\end{corollary}

\begin{proof}
When $\alpha=4$, from (\ref{eq:Psucc4}) \\$\mathcal{Z}(\tau, 4, \widehat{B}_j)=\sqrt{\tau}
\int_{\sqrt{\widehat{B}_j/\tau}}^{\infty} \frac{1}{1+u^{2}}\mathrm{d}u  = \sqrt{\tau}\arctan(\sqrt{\tau/\widehat{B}_j})$. Combining with (\ref{eq:Pk1}) and (\ref{eq:P1}) gives the desired results.
\end{proof}
When $\{\widehat{B}_j\}=1$, i.e. no biasing, the outage probability is simplified to a single trigonometric function and a single variable.
\begin{corollary}
When $\{\alpha_j\}=4$ and $\{\widehat{B}_j\}=1$, i.e. unbiased association, the outage probability of $k$th tier and the network outage probability are given as
\begin{align}\label{eq:Pk12}
\mathcal{O}(\tau, \{\lambda_j\}, 4, 1) = \mathcal{O}_{k}(\tau, \{\lambda_j\}, 4, 1)= 1-\frac{1}{1+\sqrt{\tau}\arctan(\sqrt{\tau})}.
\end{align}
\end{corollary}

\begin{proof}
When $\{\alpha_j\}=4$ and $\{\widehat{B}_j\}=1$, from (\ref{eq:Psucc4}), $\mathcal{Z}(\tau, 4, 1)=\sqrt{\tau}
\int_{1/\sqrt{\tau}}^{\infty} \frac{1}{1+u^{2}}\mathrm{d}u  = \sqrt{\tau}\arctan(\sqrt{\tau})$. Combining with (\ref{eq:Pk11}) gives the desired results.
\end{proof}
%Comment the result for K=1 and comparison with Jeff's cellular paper.
As the simplified expression of Corollary 2, this corollary also shows that the outage probability of each tier is the same for all tiers, and even it is the same as the outage probability of overall network. This implies that adding small pico and femto BSs to the macrocell network does not change the SINR distribution of each tier, because the increase in interference power is counter-balanced by the increase in signal power. We thus expect that both network sum throughput and per-tier sum throughput linearly increases with the number of BSs without any interference management techniques, since the SINR of each tier cell remains same. Furthermore, for a single-tier $K$=1, the expression of (\ref{eq:Pk12}) is the same as the result of
\cite[eq. (25)]{AndBac10}. Our work thus extends the single-tier approach \cite{AndBac10} to a general multi-tier in the special case of Rayleigh fading.

\section{Spectral Efficiency}
We derive the average ergodic rate and the minimum average user throughput to measure spectral efficiency performance of the network. The average ergodic rate is obtained with similar tools as the outage probability was in Section III. The minimum average user throughput is derived from the average ergodic rate. Both metrics are computed in units of nats/sec/Hz to simplify the expressions and analysis, with 1 nat/s $=$ 1.443 bps.
\subsection{Average Ergodic Rate}
In this section, we derive the average ergodic rate of a typical randomly located user in the $K$-tier cellular network. Using the same approach as in (\ref{eq:Pout}), the average ergodic rate is given as
\begin{align}\label{eq:Rate}
\mathcal{R}
= \sum_{k=1}^K \mathcal{R}_{k}\mathcal{A}_{k},
\end{align}
where $\mathcal{A}_{k}$ is the probability that a typical user is connected to the $k$th tier, which is given in Lemma 1. We denote $\mathcal{R}_{k}$ as the average ergodic rate of a typical user associated with $k$th tier. To compute $\mathcal{R}_{k}$, we first consider the ergodic rate of a user at a distance $x$ from its serving $k$th tier BS. The link rate then is averaged over the distance $x$ (i.e. over the $k$th tier).The average ergodic rate of the $k$th tier is thus defined as
\begin{align}\label{eq:RateDef}
\mathcal{R}_{k} &\triangleq \mathbb{E}_x\left[ \mathbb{E}_{\mathtt{SINR}_k}\left[\ln(1+\mathtt{SINR}_k(x))\right] \right]
\end{align}
The metric means the average data rate of a randomly chosen user within the $k$th tier with one active user per cell. It also represents the average cell throughput of the $k$th tier with orthogonal multiple access and round-robin scheduling.
We first derive the most general results of the average ergodic rate, considering thermal noise as well as per-tier BS density $\lambda_j$, bias factor $B_j$, and path loss exponent $\alpha_j$.
\begin{theorem}
The average ergodic rate of a typical user associated with $k$th tier is
\begin{align}\label{eq:Rk}
\mathcal{R}_{k}=\frac{2\pi\lambda_k}{\mathcal{A}_k}\int_{x>0} \int_{t>0} \exp \left\{- \tfrac{e^t -1}{\mathtt{SNR}}  -\pi \textstyle\sum_{j=1}^{K} x^{2/\widehat{\alpha}_j} C_{j}(t) \right\} x \mathrm{d}t  \mathrm{d}x,
\end{align}
where
\begin{align*}
C_{j}(t) = \lambda_j \widehat{P}_j^{2/\alpha_j} (\widehat{B}_j^{2/\alpha_j}+\mathcal{Z}(e^t -1, \alpha_j, \widehat{B}_j)).
\end{align*}
Furthermore, the average ergodic rate of overall network is
\begin{align}\label{eq:R}
\mathcal{R}=\sum_{k=1}^{K} 2\pi\lambda_k\int_{x>0} \int_{t>0} \exp \left\{- \tfrac{e^t -1}{\mathtt{SNR}}  -\pi \textstyle\sum_{j=1}^{K} x^{2/\widehat{\alpha}_j} C_{j}(t) \right\} \mathrm{d}t x \mathrm{d}x,
\end{align}
\end{theorem}
\begin{proof}See Appendix~\ref{sec:PFThm1}.
\end{proof}
Although not closed-form, this expression is efficiently computed numerically as opposed to the usual Monte Carlo methods that rely on repeated random sampling to compute their results.

We now consider the special case, where ignoring thermal noise ($W=0$), and applying unbiased association ($\{\widehat{B}_j\}=1$) and equal path loss exponents ($\{\alpha_j\}=\alpha$) for all tiers.
\begin{corollary}
When $\{\alpha_j\}=\alpha$ and $\{\widehat{B}_j\}=1$, the average ergodic rate of overall network and the average ergodic rate of a typical user associated with $k$th tier are
\begin{align}\label{eq:Rk11}
\mathcal{R}(\{\lambda_j\}, \alpha, 1) = \mathcal{R}_{k}(\{\lambda_j\}, \alpha, 1)=\int _{t=0}^\infty \frac{1}{1+\mathcal{Z}(e^t-1, \alpha,1)}\mathrm{d}t.
\end{align}
If $\{\alpha_j\}=4$, the average rates are further simplified to
\begin{align}\label{eq:Rk12}
\mathcal{R}(\{\lambda_j\}, 4, 1) = \mathcal{R}_{k}(\{\lambda_j\}, 4, 1)= \int_{t=0}^\infty \frac{1}{1+\sqrt{e^t-1}\arctan(\sqrt{e^t-1})}\mathrm{d}t.
\end{align}
\end{corollary}
\begin{proof}
If $\{\alpha_j\}=\alpha$ and $\{\widehat{B}_j\}=1$, from (\ref{eq:Rk}) we obtain
\begin{align}\label{eq:Rk3}
\mathcal{R}_{k}(\{\lambda_j\}, \alpha, 1)=\frac{2\pi\lambda_k}{\mathcal{A}_k}\int_{t>0} \int_{x>0} x \exp\left\{ -\pi \textstyle\sum_{j=1}^{K} \lambda_j \widehat{P}_j^{2/\alpha} (1+\mathcal{Z}(e^t -1, \alpha, 1)) x^2  \right\}  \mathrm{d}x \mathrm{d}t,
\end{align}
where $\frac{\lambda_k}{\mathcal{A}_k}=\sum_{j=1}^{K} \lambda_j \widehat{P}_j^{2/\alpha}$ from (\ref{eq:Lem1-4}).
Employing a change of variables $x^2=v$ and $\int_0^\infty e^{-Av}\mathrm{d}v=\frac{1}{A}$ give
\begin{align}\label{eq:Rk13}
\mathcal{R}_{k}(\{\lambda_j\}, \alpha, 1)=\int _{t=0}^\infty \frac{1}{1+\mathcal{Z}(e^t-1, \alpha,1)}\mathrm{d}t,
\end{align}
which is independent on the index $k$. From (\ref{eq:Rate}) and $\sum_{k=1}^K \mathcal{A}_k=1$, we obtain
$\mathcal{R} = \mathcal{R}_{k}\sum_{k=1}^K \mathcal{A}_k=\mathcal{R}_{k}$. This gives the result in (\ref{eq:Rk11}). When $\alpha=4$ and $\{\widehat{B}_j\}=1$, from (\ref{eq:Psucc4}), $\mathcal{Z}(e^t-1, 4, 1)=\sqrt{\tau}
\int_{1/\sqrt{(e^t-1)}}^{\infty} \frac{1}{1+u^{2}}\mathrm{d}u  = \sqrt{(e^t-1)}\arctan(\sqrt{e^t-1})$. Combining with (\ref{eq:Rk11}) gives the desired results.
\end{proof}
In this corollary, the double integration in Theorem 2 is simplified to a single integration, and the integrand is especially simple to compute. The average ergodic rate, like the outage probability, is not affected by BS transmit power $P_j$, BS density $\lambda_j$, and the number of tiers $K$. This means that adding BSs or raising the power increases interference and desired signal power by the same amount, and they offset each other. Therefore, the network sum rate increases in direct proportion to the total number of BSs.

\subsection{Minimum Average User Throughput}
We assume the orthogonal transmission, where equal time (and/or frequency) slots are allocated to each user one after the other in a round-robin manner. The average ergodic rate of the $k$th-tier user $\mathcal{R}_k$ then means the average cell throughput in the $k$th tier. For the $k$th tier, the average user throughput of a cell is given as
\begin{align}\label{eq:AvgRk}
\mathcal{\overline{R}}_{k} = \frac{\mathcal{R}_k}{\mathcal{N}_k},
\end{align}
where $\mathcal{N}_k$ is the average number of user per cell of the $k$th tier, which is given in Lemma 2. Combining this equation with (\ref{eq:Cor1-3}) and (\ref{eq:Rk}) gives
\begin{align}\label{eq:Rbark}
\mathcal{\overline{R}}_{k}=\frac{2\pi\lambda_k^2}{\mathcal{A}_k^2 \lambda^{\mathrm{(u)}}}\int_{x>0} \int_{t>0} \exp \left\{- \tfrac{e^t -1}{\mathtt{SNR}}  -\pi \textstyle\sum_{j=1}^{K} x^{2/\widehat{\alpha}_j} C_{j}(t) \right\} \mathrm{d}t x \mathrm{d}x.
\end{align}
We define the minimum average user throughput as
\begin{align}\label{eq:Q}
\mathcal{Q} \triangleq \min_{k\in\{1,\cdots,K\}} \mathcal{\overline{R}}_{k},
\end{align}
which takes the minimum value among the $K$ values given by (\ref{eq:AvgRk}). The metric represents the minimum quality of service (QoS) that the network can provide. Since it is highly dependent on the number of user in a cell, i.e. cell load, the minimum average user throughput well measures the effect of biasing on the QoS of the HCN.

\section{Numerical Results}
\subsection{Accuracy of Model and Analysis}
We use a path loss at 1 meter of $L_0=-38.5$ dB and thermal noise $W=-104$ dBm (i.e. 10 MHz bandwidth) for all numerical results. We obtain the outage probability (SINR CDF) using Monte Carlo methods where BSs are deployed according to the given model, and a user is fixed at the origin for each network realization.  For each spatial realization, a SINR sample is obtained by generating independent Rayleigh random variables for the fading channels.

Fig. \ref{fig:comparison} compares the outage probability of the proposed PPP BS deployment, an actual tier-1 BS deployment, and a hexagonal grid model for tier 1.  A total of three tiers are modeled with the lower two tiers (e.g. pico and femto BSs) modeled according to a PPP. We observe that the tier-1 PPP model is nearly as accurate as a hexagonal grid model, where the grid model provides lower bound with a gap less than 1 dB from actual BS deployment and the PPP model gives upper bound with less than a 1.5 dB gap from an actual BS deployment. Similar results are also observed for two-tier case in \cite{DhiAndJ11}.
The analytic curves given from (\ref{eq:P}) are remarkably close to the simulated curves for all considered SINR threshold, which is an advantage over \cite{DhiAndJ11} that provides an exact expression for $\tau > 1$.

\subsection{Effect of BS Density, Path Loss Exponent, and Biasing}
We obtained numerical results of outage probability (in Theorem 1), average ergodic rate (in Theorem 2), minimum average user throughput (from (\ref{eq:Rbark}) and (\ref{eq:Q})) with respect to main network parameters; bias factor, BS density, and path loss exponent. Although all results given in this section are for a two-tier HCN (e.g. macrocell and picocell), they can be applied, without loss of generality, for a general $K$-tier HCN. Biasing effect is investigated by adjusting the bias factor of the picocell with no biasing of the macrocell. Since we assume the transmit power of the macro BS 20 dB larger than pico BS, the 20 dB picocell bias factor means the cell association where the user connects to the BS with the lowest path loss (or to the nearest BS for the same path loss exponents).

\textbf{Outage Probability and Average Ergodic Rate for BS density.} In Fig. \ref{fig:DensityEffect}, although considering noise, we observe no changes in outage and rate for adding BS with different power, no biasing, and same path loss exponents. This means that inter-BS interference is still dominant in our two-tier scenario considering a typical value of BS density and BS transmit power, and Corollaries 2 and 5  hold for the HCN. When the low-tier (pico or femto) BSs experience higher path loss, the outage and average rate improves as the BSs are added. This is even more optimistic than the result for the same path loss exponents. Intuitively, higher path loss reduces the interference between picocell and macrocell so the picocell is more isolated from the macrocell network. The results indicate that if given the choice, new BSs are better deployed in an area with higher path loss.

\textbf{Outage probability and Average Ergodic Rate for Biasing.} Fig. \ref{fig:BiasEffect} shows the effect of bias factor and BS density. We apply the same path loss exponent for all tiers. Unbiased association results in the same outage/rate for all BS density as provided in Corollaries 2 and 5. As the picocell bias factor increases, more macro users with low SINR are associated with the picocell, which improves the outage and rate of picocell, but degrades those of the macrocell. In terms of the outage and rate of the overall network, unbiased association always outperforms biasing. Intuitively, in biased association, some users are associated with the BS not offering the strongest received signal, which reduces the SINR of the users. The results are from the condition that every cell is fully traffic-loaded (full queues at all times). Note that for a lightly-loaded HCN, biasing can improve the rate over the whole network even if it reduces the SINR of the users, since their share of the total resources (typically, time and frequency slots) will increase.

%lambda effect
In Fig. \ref{fig:BiasEffect}, for a given bias factor, deploying more pico BSs enhances the outage and rate of macrocell (which means a decrease in the outage probability and an increase in the average ergodic rate), because more macro users with low SINR and at cell edge become associated with the picocell. The outage and rate of picocell are also improved, since adding more pico BSs tends to reduce each picocells coverage area, despite the biasing.  Interestingly, the BS density only slightly changes the outage of overall network. Although adding pico BSs increases the interference to macrocells, it also decreases their association probability. These two competing effects more or less cancel.  Similarly, the rate is also not strongly affected by the BS density.

\textbf{Minimum Average User Throughput for Biasing.} In Fig. \ref{fig:MinRate}, as the bias-factor increases, the minimum user rate increases at first but then decreases for a sufficiently large bias factor. Although the macrocell load decreases at the cost of an increase in picocell load, the average user throughput of picocell is still higher than that of macrocell for a smaller bias factor. For a sufficiently large bias factor, the average user throughput in the picocell is lower than that of the macrocell due to a massive number of connections to the picocell. Again, we wish to emphasize that these observed trends on biasing can be considered preliminary, and are heavily dependent on the cell loading. Ideally, it seems that a network should dynamically push users onto lightly loaded cells as a function of the current network conditions, and that in general, small cells will be more lightly loaded unless proactive biasing is introduced.

\section{Conclusions}
This paper developed a new analytical framework for evaluating outage probability and spectral efficiency in HCNs with flexibly cell association, also known as biasing.  It is interesting to observe that the number of tiers and density of base stations at most weakly affects the outage probability and the average ergodic rate, and under certain assumptions, does not affect them. This implies that even randomly adding pico and femtocells to a network for capacity improvement need not decrease the quality of the network, as is commonly feared. Assuming full queues at all base stations, biasing deteriorates the outage and rate of the overall network by lowering the SINR, but further work on this topic is needed, perhaps with the assistance of the model and results developed in this paper.

\section{Acknowledgments}
The authors thank Amitava Ghosh and Bishwarup Mondal of Motorola Solutions (Recently acquired by Nokia Siemens) for their advice regarding the system model and parameters, and H. S. Dhillon for his help in making Fig. \ref{fig:Biasing}.

\appendix[]
\subsection{Proof of Lemma 1}\label{sec:PFLem1}
Denote $n$ as an index of tier associating the typical user.
When $P_{\mathrm{r},k} > P_{\mathrm{r},j}$ for all $j\in \{1 \cdots K\}, j\neq k$, a typical user is associated with the $k$th tier, i.e. $n=k$. Therefore,
\begin{align}\label{eq:Lem1-1}
\mathcal{A}_k&\triangleq\mathbb{P}\left[n=k\right] \cr
&= \mathbb{E}_{R_k}\left[\mathbb{P}\left[P_{\mathrm{r},k}(R_k)> \max_{j, j\neq k} P_{\mathrm{r},j} \right] \right] \cr
&= \mathbb{E}_{R_k}\left[\prod_{j=1, j\neq k}^K \mathbb{P}\left[P_{\mathrm{r},k}(R_k) > P_{\mathrm{r},j}\right]\right]\cr
&\mathop =\limits^{\left( a \right)}
\mathbb{E}_{R_k}\left[\prod_{j=1, j\neq k}^K \mathbb{P}\left[R_j > \left(\tfrac{P_j}{P_k}\tfrac{B_j}{B_k}\right)^{1/\alpha_j}R_k^{\alpha_k/\alpha_j} \right]\right]\cr
&=\int_0^{\infty} {\textstyle\prod}_{j=1, j\neq k}^K \mathbb{P}\left[R_j > (\widehat{P}_j\widehat{B}_j)^{1/\alpha_j}r^{\widehat{\alpha}_j} \right] f_{R_k}(r)\mathrm{d}r,
\end{align}
where $(a)$ is given using (\ref{eq:Pr_j}). $\mathbb{P}[R_j > (\widehat{P}_j \widehat{B}_j)^{1/\alpha_j}r^{\widehat{\alpha}_j} ]$ and the PDF of $R_k$ are derived using the null probability of a 2-D Poisson process with density $\lambda$ in an area $A$, which is $\exp(-\lambda A)$.
\begin{align}\label{eq:Lem1-2}
\prod_{j=1, j\neq k}^K \mathbb{P}\left[R_j > (\widehat{P}_j \widehat{B}_j)^{1/\alpha_j}r^{\widehat{\alpha}_j} \right]&=\prod_{j=1, j\neq k}^K\mathbb{P}[\mathrm{No~ BS~closer ~ than~} (\widehat{P}_j \widehat{B}_j)^{1/\alpha_j}r^{\widehat{\alpha}_j} \mathrm{~in~the~} j\mathrm{~th ~tier}]\cr
&=\prod_{j=1, j\neq k}^K e^{-\pi \lambda_j (\widehat{P}_j\widehat{B}_j)^{2/\alpha_j}r^{2/\widehat{\alpha}_j}}.
\end{align}
and
\begin{align}\label{eq:Lem1-3}
f_{R_k}(r)=1-\frac{\mathrm{d}\mathbb{P}[R_k>r]}{\mathrm{d}r}=e^{-\pi\lambda_k r^2}2\pi \lambda_k r.
\end{align}
Combining (\ref{eq:Lem1-1}), (\ref{eq:Lem1-2}), and (\ref{eq:Lem1-3}), we obtain
\begin{align}\label{eq:Lem1-3-1}
\mathcal{A}_k = 2\pi \lambda_k \int_0^{\infty} r \exp \left\{-\pi \textstyle\sum_{j=1, j\neq k}^{K} \lambda_j (\widehat{P}_j \widehat{B}_j)^{2/\alpha_j} r^{2/\widehat{\alpha}_j} -\pi \lambda_k r^2 \right \} \mathrm{d}r.
\end{align}
Since $\widehat{P}_j=1$, $\widehat{B}_j=1$, and $\widehat{\alpha}_j=1$ for $j=k$, we obtain $\sum_{j=1, j\neq k}^{K} \lambda_j (\widehat{P}_j \widehat{B}_j)^{2/\alpha_j} r^{2/\widehat{\alpha}_j} + \lambda_k r^2= \sum_{j=1}^{K} \lambda_j (\widehat{P}_j \widehat{B}_j)^{2/\alpha_j} r^{2/\widehat{\alpha}_j}$. From (\ref{eq:Lem1-3-1}), we thus obtain the desired result in (\ref{eq:Lem1}).

If $\{\alpha_j\}=\alpha$, then we obtain $\widehat{\alpha}_j=1$ for all $j\in\{1,\cdots,K\}$. This gives
\begin{align}
\mathcal{A}_k = 2\pi \lambda_k \int_0^{\infty} r \exp \left\{-\pi \textstyle\sum_{j=1}^{K} \lambda_j (\widehat{P}_j \widehat{B}_j)^{2/\alpha} r^{2} \right\} \mathrm{d}r.
\end{align}
Employing the change of variables $r^2=t$ gives
\begin{align}\label{eq:Lem1-4}
\mathcal{A}_k &= \frac{\lambda_k} {\sum_{j=1}^{K} \lambda_j (\widehat{P}_j \widehat{B}_j)^{2/\alpha}} \cr
&= \frac{\lambda_k} {\sum_{j=1,j\neq k}^{K} \lambda_j (\widehat{P}_j \widehat{B}_j)^{2/\alpha} + \lambda_k}.
\end{align}
By applying $\widehat{P}_j\triangleq\frac{P_j}{P_k}$ and $\widehat{B}_j\triangleq\frac{B_j}{B_k}$ to (\ref{eq:Lem1-4}), we obtain the simple result in (\ref{eq:Lem1-0}).

%After some algebra, we obtain the desired result.
\subsection{Proof of Theorem 1}\label{sec:PFThm1}
From (\ref{eq:PoutDef}), the outage probability of the $k$th tier is given as
\begin{align}\label{eq:Thm1-Pout}
\mathcal{O}_k
&= 1-\int_{x=0}^\infty \mathbb{P}[\mathtt{SINR}_k(x) > \tau]f_{X_k}(x )\mathrm{d}x
\cr
&= 1-\frac{2\pi\lambda_k }{\mathcal{A}_k}\int_{x=0}^\infty \mathbb{P}[\mathtt{SINR}_k(x) > \tau] x \exp \left\{-\pi\textstyle\sum_{j=1}^{K} \lambda_j (\widehat{P}_j \widehat{B}_j)^{2/\alpha_j} x^{2/\widehat{\alpha}_j} \right \} \mathrm{d}x
,
\end{align}
where $f_{X_k}(x)$ is given in Lemma 2.
The user SINR in (\ref{eq:SINR}) is rewritten as $\gamma(x) = \frac{g_{k,0}}{x^{\alpha_k}P_k^{-1} Q}$,
%\begin{equation}\label{eq:Lem2-1}
%\gamma(R) = \frac{g_0}{K (I_1 +I_2)},
%\end{equation}
where $Q=\sum_{j=1}^K I_j+W/L_0$
The CCDF of the user SINR at distance $x$ from its associated BS in $k$th tier is given as
\begin{align}\label{eq:Psucc1}
\mathbb{P}[\gamma_k(x) > \tau]&=\mathbb{P}[g_{k,0} > x^{\alpha_k}P_k^{-1} \tau Q] \cr
&=\int_0^\infty \exp \left\{-x^{\alpha_k}P_k^{-1} \tau q \right\} f_Q(q) \mathrm{d}q\cr
&= \mathbb{E}_Q[\exp \left\{-x^{\alpha_k}P_k^{-1} \tau q \right\}]\cr
&= \exp \left\{-\frac{\tau}{\mathtt{SNR}}\right\}\prod_{j=1}^K \mathcal{L}_{I_j}(x^{\alpha_k}P_k^{-1} \tau )
%&= \mathcal{L}_{W}(x^{\alpha_k}P_k^{-1} \tau )\prod_{j=1}^K \mathcal{L}_{I_j}(x^{\alpha_k}P_k^{-1} \tau)\cr
%&= e^{-\frac{\sigma^2(x^{\alpha_k}P_k^{-1} \tau )^2}{2}}\prod_{j=1}^K \mathcal{L}_{I_j}(x^{\alpha_k}P_k^{-1} \tau),
\end{align}
where $\mathtt{SNR}=\frac{P_k L_0 x^{-\alpha_k}}{W}$, and the Laplace transform of $I_j$ is
\begin{align}\label{eq:Psucc2}
\mathcal{L}_{I_j}(x^{\alpha_k}P_k^{-1} \tau )&=\mathbb{E}_{I_j}[e^{-x^{\alpha_k}P_k^{-1} \tau I_j}]\cr
&=\mathbb{E}_{\Phi_j}\left[\exp\left\{-x^{\alpha_k}\widehat{P}_j \tau \sum_{i \in \Phi_j } h_{j,i} |Y_{j,i}|^{-\alpha_j}\right\}\right]\cr
&\mathop =\limits^{\left( a \right)}
\exp \left\{ -2\pi\lambda_j\int_{z_j}^\infty\left(1-\mathcal{L}_{h_j}(x^{\alpha_k}\widehat{P}_j \tau y^{-\alpha_j} )\right) y \mathrm{d}y\right\}
\cr
&\mathop =\limits^{\left( b \right)}
\exp \left\{ -2\pi\lambda_j\int_{z_j}^\infty\left(1-\frac{1}{1+ x^{\alpha_k}\widehat{P}_j \tau y^{-\alpha_j}}\right) y \mathrm{d}y\right\}\cr
&=\exp\left\{-2\pi\lambda_j\int_{z_j}^\infty \frac{y}{1+(x^{\alpha_k} \widehat{P}_j \tau )^{-1}  y^{\alpha_j}} \mathrm{d}y\right\},
\end{align}
where (a) is given from \cite{AndBac10}, and (b) follows because interference fading power $h_j \sim \exp(1)$. The integration limits are from $z_j$ to $\infty$ since the closest interferer in $j$th tier is at least at a distance $z_j=(\widehat{P}_j\widehat{B}_j)^{1/\alpha_j}x^{\widehat{\alpha}_j}$. Employing a change of variables $u=(x^{\alpha_k} \widehat{P}_j \tau )^{-2/\alpha_j}y^2$ results in
\begin{align}\label{eq:Psucc3}
\mathcal{L}_{I_j}(x^{\alpha_k}P_k^{-1} \tau )=\exp \left\{ -\pi\lambda_j \widehat{P}_j ^{2/\alpha_j} \mathcal{Z}(\tau, \alpha_j, \widehat{B}_j) x^{2/\widehat{\alpha}_j} \right\},
\end{align}
where
\begin{align}\label{eq:Psucc4}
\mathcal{Z}(\tau, \alpha_j, \widehat{B}_j)&=\tau^{\frac{2}{\alpha_j}}
\int_{( \widehat{B}_j / \tau )^{2/\alpha_j}}^{\infty} \frac{1}{1+u^{\alpha_j/2}}du \cr
&=\frac{2\tau \widehat{B}_j^{2/\alpha_j-1}}{\alpha_j-2} ~_{2}F_1\left[1, 1\!-\!\tfrac{2}{\alpha_j}; 2\!-\!\tfrac{2}{\alpha_j};  -\tfrac{\tau} {\widehat{B}_j} \right] \mathrm{~for~} \alpha_j>2
\end{align}
Here, $_2F_1[\cdot]$ denotes the Gauss hypergeometric function.
Plugging (\ref{eq:Psucc3}) into (\ref{eq:Psucc1}) gives
\begin{equation}\label{eq:Thm1-Psucc}
\mathbb{P}[\gamma_k(x) > \tau] =  \exp \left\{-\frac{\tau }{\mathtt{SNR}} - \pi \sum_{j=1}^{K} \lambda_j  \widehat{P}_j^{2/\alpha_j}  \mathcal{Z}(\tau, \alpha_j, \widehat{B}_j) x^{2/\widehat{\alpha}_j} \right \}.
\end{equation}
Combining (\ref{eq:Thm1-Pout}) and (\ref{eq:Thm1-Psucc}), we obtain the per-tier outage probability in (\ref{eq:Pk}). Furthermore, plugging (\ref{eq:Pk}) into (\ref{eq:Pout}) gives the network outage probability in (\ref{eq:P}).

\subsection{Proof of Theorem 2}\label{sec:PFThm1}
From (\ref{eq:RateDef}), the average ergodic rate of the $k$th tier is
\begin{align}\label{eq:Thm2-Rate}
\mathcal{R}_{k}
&=\int_{0}^{\infty} \mathbb{E}_{\mathtt{SINR}_k}\left[\ln(1+\mathtt{SINR}_k(x))\right] f_{X_k}(x)\mathrm{d}x\cr
&=\frac{2\pi\lambda_k }{\mathcal{A}_k}\int_{0}^{\infty} \mathbb{E}_{\mathtt{SINR}_k}\left[\ln(1+\mathtt{SINR}_k(x))\right] x \exp \left\{-\pi\textstyle\sum_{j=1}^{K} \lambda_j (\widehat{P}_j \widehat{B}_j)^{2/\alpha_j} x^{2/\widehat{\alpha}_j} \right \} \mathrm{d}x,
\end{align}
where $f_{X_k}(x)$ is given in Lemma 2.
Since $\mathbb{E}[X]=\int_0^{\infty}
\mathbb{P}[X>x]dx$ for $X>0$, we obtain
\begin{align}\label{eq:Thm2-1}
\mathbb{E}_{\mathtt{SINR}_k}\left[\ln(1+\mathtt{SINR}_k(x))\right]
&=\int_{0}^{\infty} \mathbb{P}\left[\ln(1+\mathtt{SINR}_k(x))>t \right] \mathrm{d}t \cr
&=\int_{0}^{\infty} \mathbb{P}\left[ g_{k,0} > x^{\alpha_k}P_k^{-1}Q (e^t-1)\right] \mathrm{d}t
\cr
&\mathop = \limits^{\left( a \right)}
\int_{0}^{\infty} e^{-\frac{e^t-1}{\mathtt{SNR}}}\prod_{j=1}^K \mathcal{L}_{I_j}(x^{\alpha_k}P_k^{-1} (e^t-1)) \mathrm{d}t,
\end{align}
where $(a)$ follows from plugging $\tau = e^t-1$ in (\ref{eq:Psucc1}). From (\ref{eq:Psucc3}), we obtain
\begin{align}\label{eq:Thm2-2}
\mathcal{L}_{I_j}(x^{\alpha_k}P_k^{-1} (e^t-1) )=\exp \left\{ -\pi\lambda_j \widehat{P}_j ^{2/\alpha_j} \mathcal{Z}(e^t-1, \alpha_j, \widehat{B}_j) x^{2/\widehat{\alpha}_j} \right\},
\end{align}
with
\begin{align}\label{eq:Thm2-3}
\mathcal{Z}(e^t-1, \alpha_j, \widehat{B}_j)&=(e^t-1)^{\frac{2}{\alpha_j}}
\int_{( \widehat{B}_j/(e^t-1)  )^{2/\alpha_j}}^{\infty} \frac{1}{1+u^{\alpha_j/2}}du \cr
&=\frac{2(e^t-1) \widehat{B}_j^{2/\alpha_j-1}}{\alpha_j-2} ~_{2}F_1\left[1, 1\!-\!\tfrac{2}{\alpha_j}; 2\!-\!\tfrac{2}{\alpha_j};  \tfrac{1-e^t}{\widehat{B}_j} \right].
\end{align}
Plugging (\ref{eq:Thm2-2}) into (\ref{eq:Thm2-1}) gives
\begin{align}\label{eq:Thm2-4}
\mathbb{E}_{\mathtt{SINR}_k}\left[\ln(1+\mathtt{SINR}_k(x))\right]=\int_{0}^{\infty} \exp \left\{-\frac{e^t-1}{\mathtt{SNR}}-\pi\sum_{j=1}^K \lambda_j \widehat{P}_j ^{2/\alpha_j} x^{2/\widehat{\alpha}_j} \mathcal{Z}(e^t-1, \alpha_j, \widehat{B}_j) \right\}  \mathrm{d}t.
\end{align}

Combining (\ref{eq:Thm2-Rate}) and (\ref{eq:Thm2-4}), we obtain the average ergodic rate of the $k$th tier in (\ref{eq:Rk}). Furthermore, plugging (\ref{eq:Rk}) into (\ref{eq:Rate}) gives the average ergodic rate of entire network in (\ref{eq:R}).
\bibliographystyle{IEEEtran}
\bibliography{IEEEabrv,HJO}
%\bibliography{IEEEabrv,HJO}
\newpage
\begin{figure}
\begin{center}
   \includegraphics[width=5in]{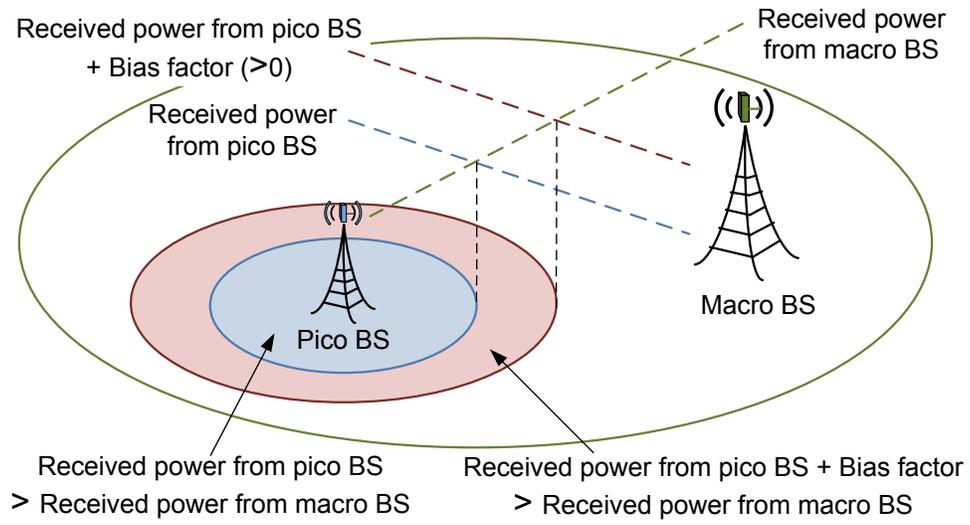}
    \caption{Concept of flexible cell association with positive bias factor in a two-tier macro/pico HCN.}
    \label{fig:BiasingView}
\end{center}
\end{figure}

\begin{figure}
\begin{center}
   \includegraphics[width=3.5 in]{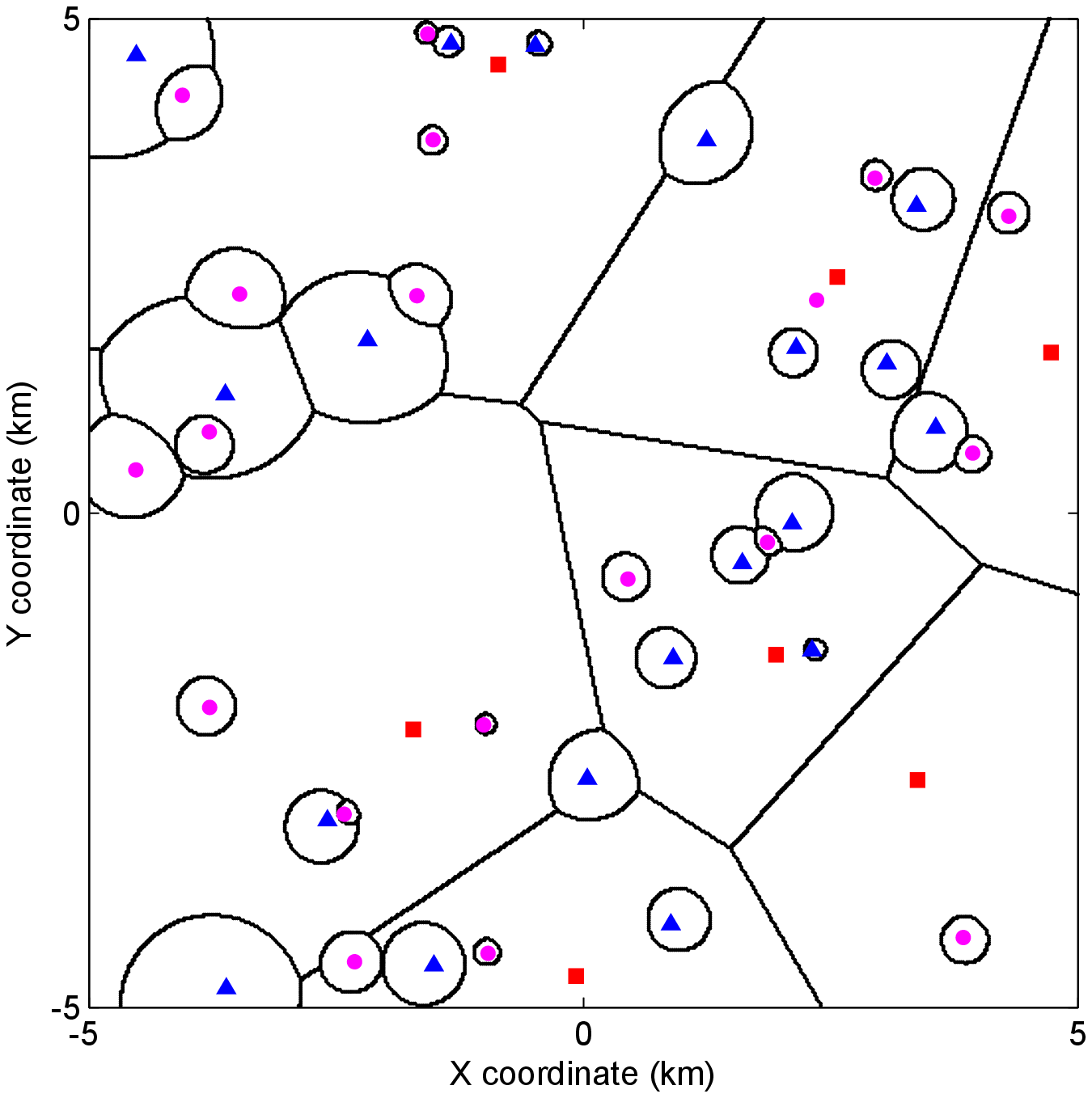}\\\vspace{0.4 in}
   \includegraphics[width=3.5 in]{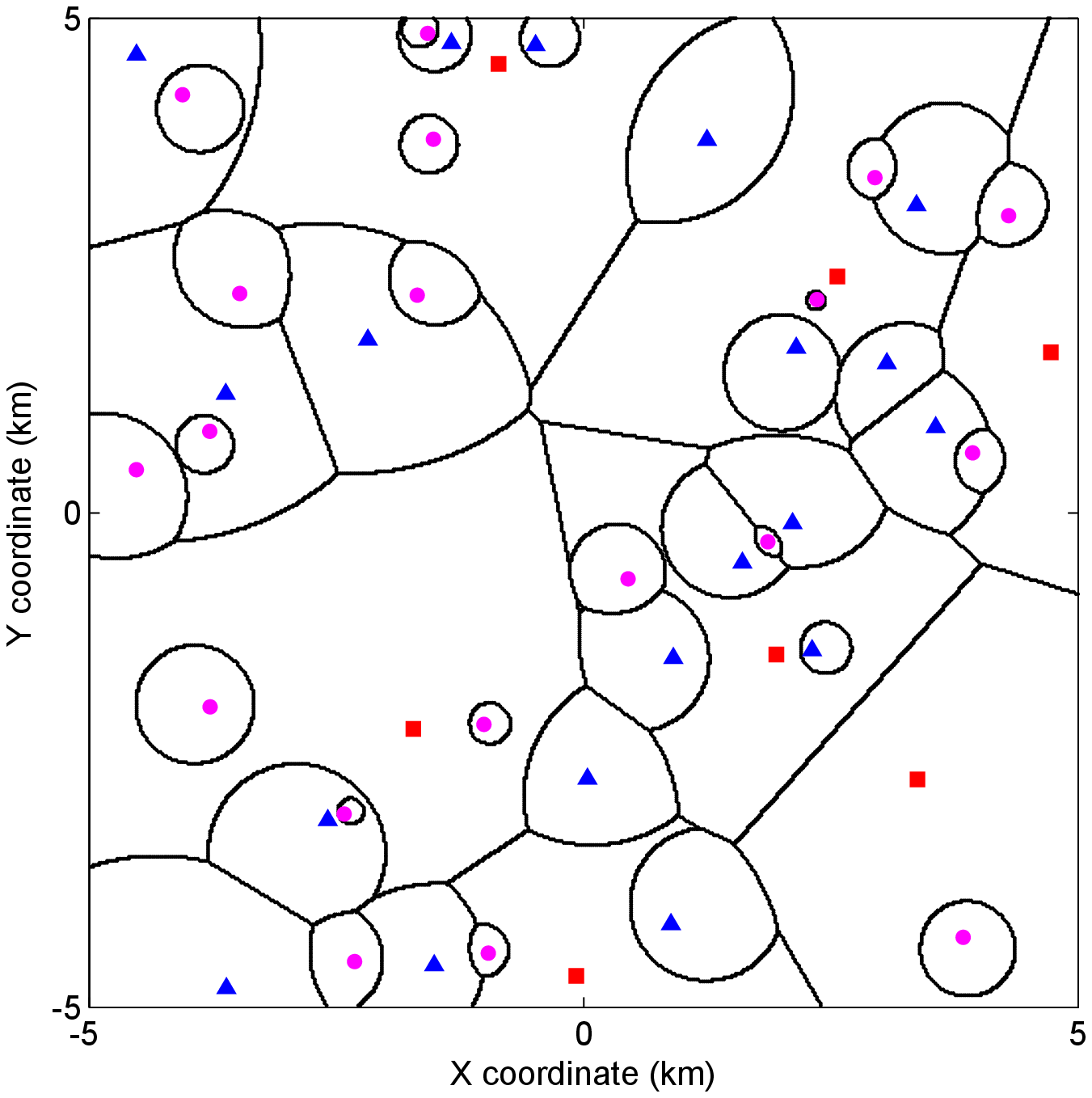}
    \caption{Example of downlink HCNs with three tiers of BSs: a high-power macrocell BS (red square) is overlaid with successively denser and lower power picocells (blue triangle) and femtocells (magenta circle). Both figures have the same BS position. Black lines show the border of cell coverage. The pico/femtocell with positive biasing (lower figure) gives a larger coverage than no biasing (upper figure). }
    \label{fig:Biasing}
\end{center}
\end{figure}

\begin{figure}
\begin{center}
   \includegraphics[width=4 in]{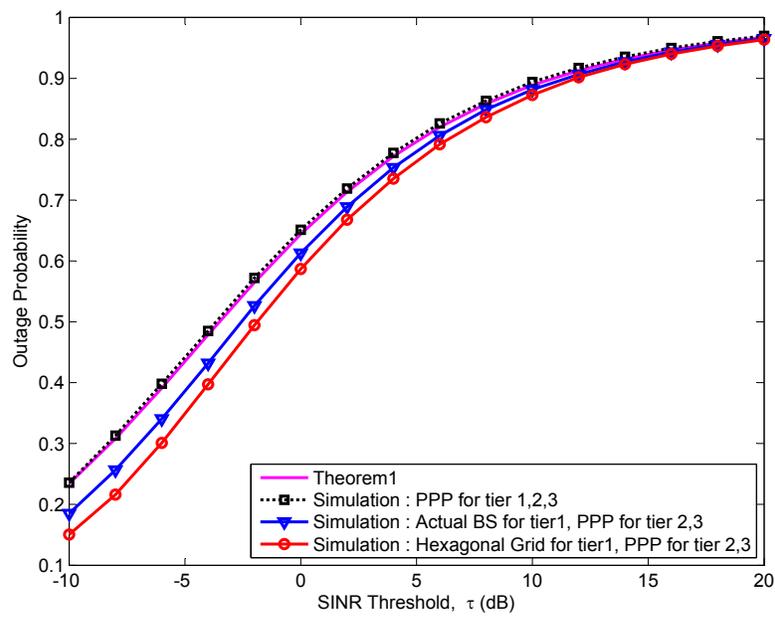}
    \caption{Outage probability comparison in a three-tier HCN ( $\lambda_2=2\lambda_1,\lambda_3=20\lambda_1$, $\{B_1, B_2, B_3\}=\{1,1,1\}$, $\{\alpha_1, \alpha_2, \alpha_3\}=\{3.8, 3.5, 4\}$, and $\{P_1, P_2, P_3\}=\{53, 33, 23\}$ dBm, ).}
    \label{fig:comparison}
\end{center}
\end{figure}

\begin{figure}
\begin{center}
   \includegraphics[width=4 in]{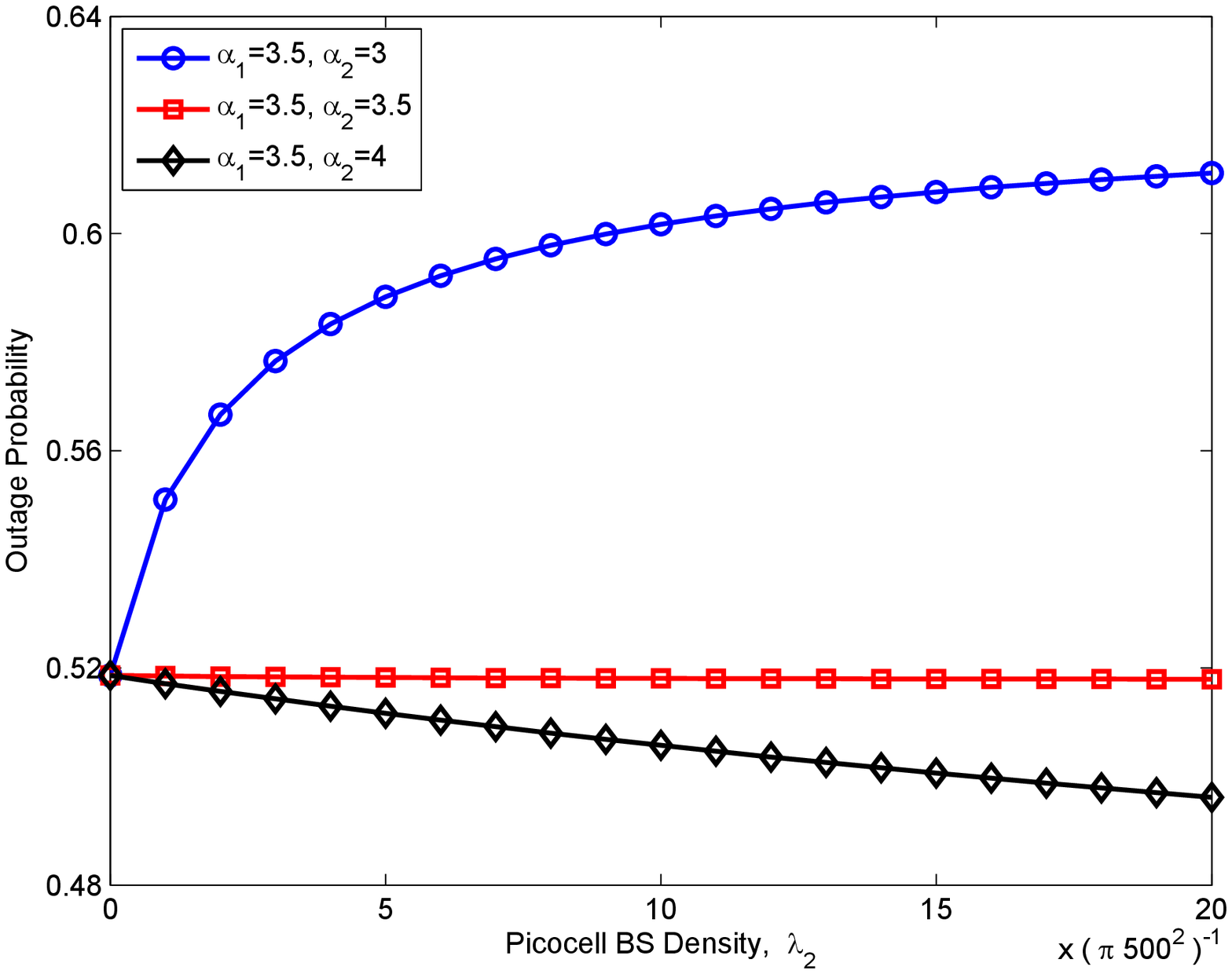}\\\vspace{0.4 in}
   \includegraphics[width=4 in]{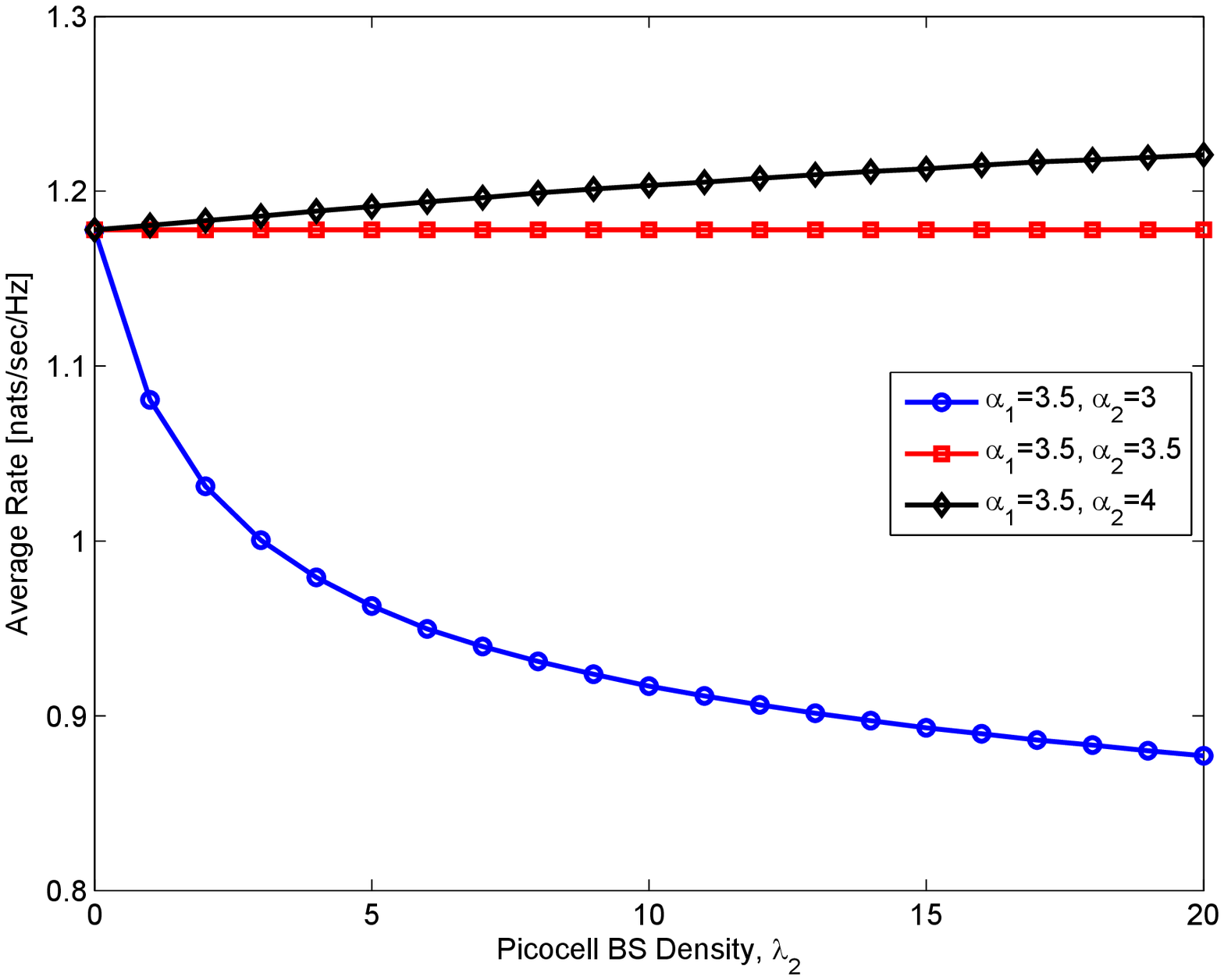}
   \caption{Outage probability and average ergodic rate for varying density of picocells in a two-tier HCN ($K=2$, $\lambda_1=\frac{1}{\pi 500^2}$, $\{B_1, B_2\}=\{1,1\}$, and $\{P_1, P_2\}=\{53, 33\}$ dBm).}
    \label{fig:DensityEffect}
\end{center}
\end{figure}

\begin{figure}
\begin{center}
   \includegraphics[width=4 in]{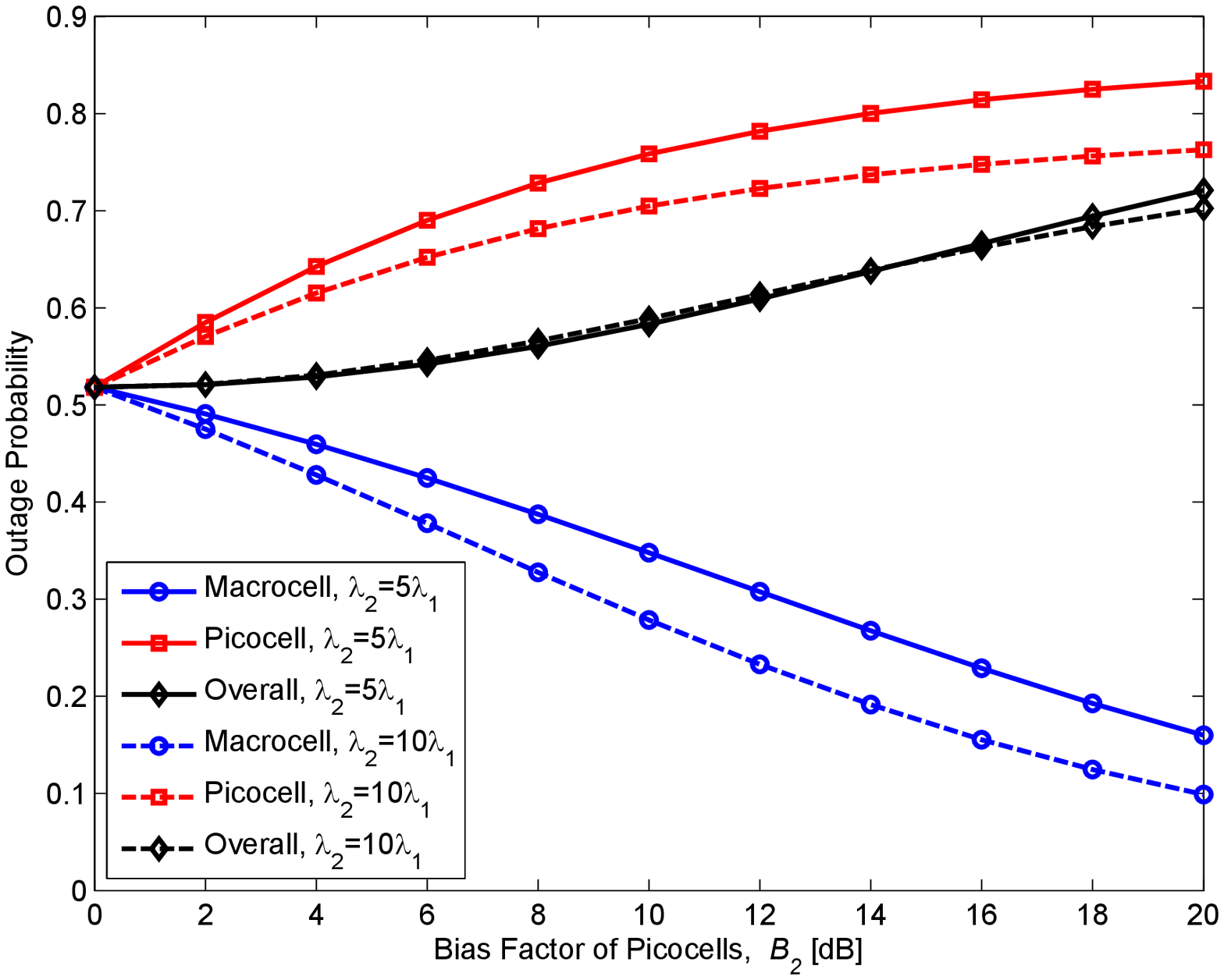}\\\vspace{0.4 in}
   \includegraphics[width=4 in]{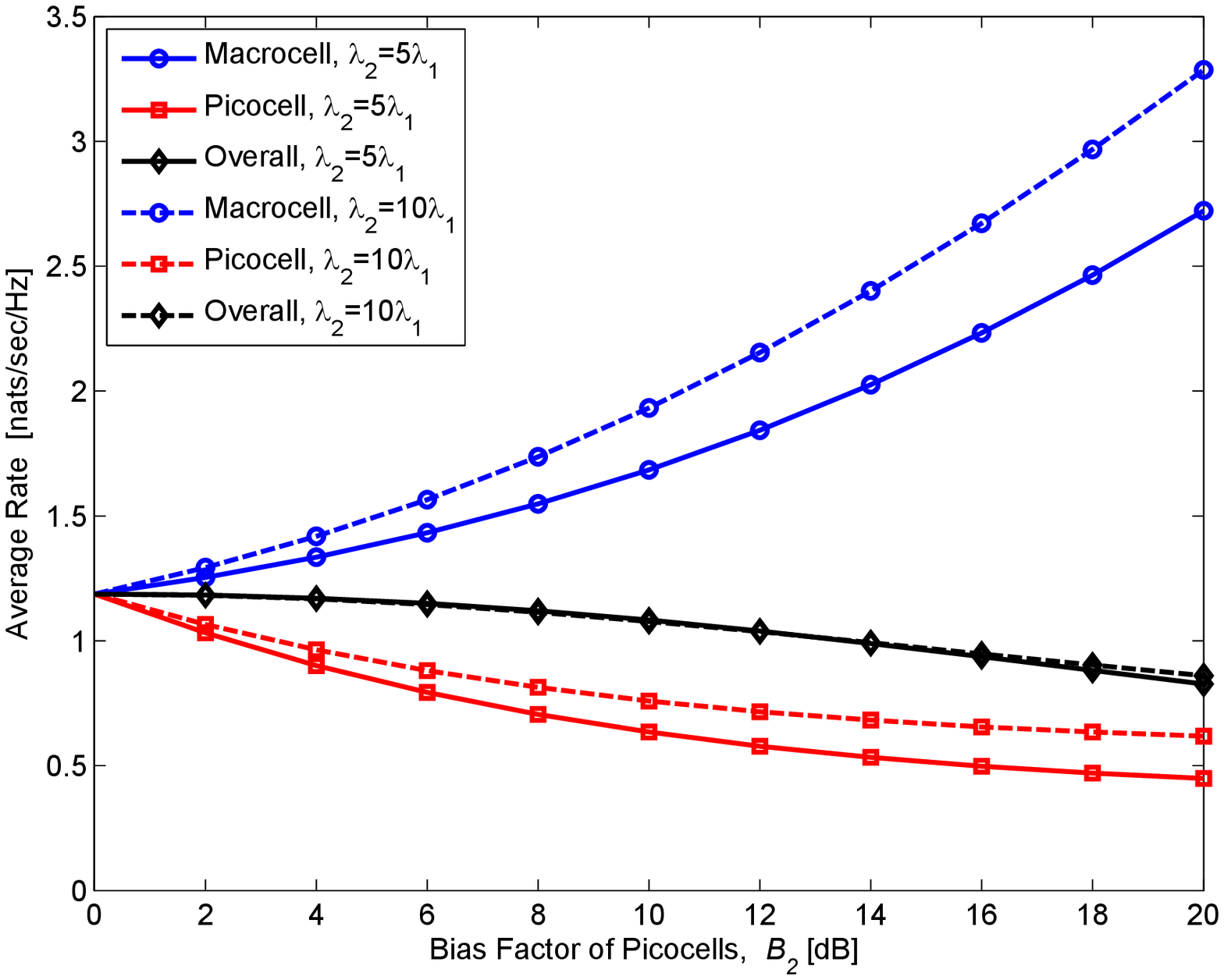}
    \caption{Outage probability and average ergodic rate for varying bias factor of picocells in a two-tier HCN ($K=2$, $B_1=1$, $\lambda_1=\frac{1}{\pi 500^2}$, $\{\alpha_1, \alpha_2\}=\{3.5, 3.5\}$, and $\{P_1, P_2\}=\{53, 33\}$ dBm).}
    \label{fig:BiasEffect}
\end{center}
\end{figure}

\begin{figure}
\begin{center}
   \includegraphics[width=4 in]{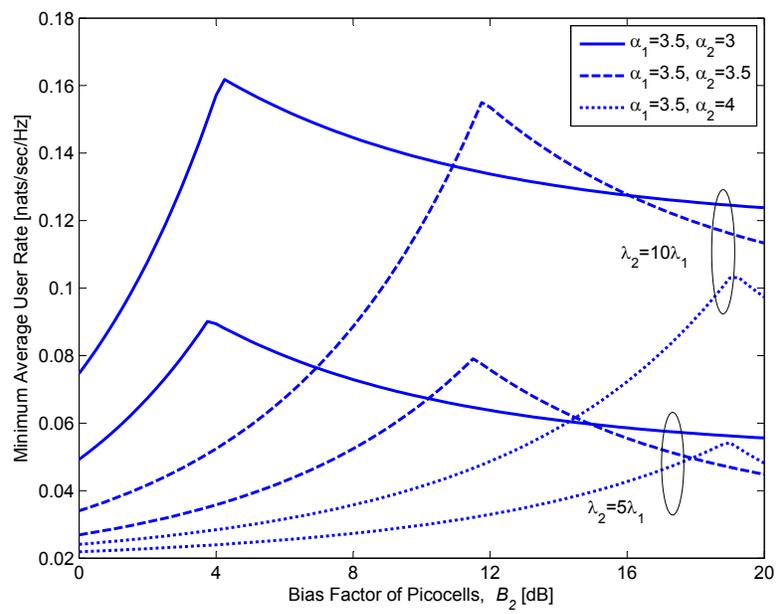}
    \caption{Minimum average user throughput for varying bias-factor and density of picocells in a two-tier HCN ($K=2$, $B_1=1$, $\lambda_1=\frac{1}{\pi 500^2}$, $\{P_1, P_2\}=\{53, 33\}$ dBm).}
    \label{fig:MinRate}
\end{center}
\end{figure}
\end{document}